\documentclass[11pt,a4paper,reqno]{amsart}

\usepackage[T1]{fontenc}
\usepackage{lmodern}
\usepackage[UKenglish]{babel}
\usepackage[margin=1in]{geometry}
\usepackage{amsmath,amssymb,amsthm,mathtools}
\usepackage{enumitem}
\usepackage{xcolor}
\usepackage{tikz-cd}
\usepackage[unicode]{hyperref}

\theoremstyle{plain}
\newtheorem{theorem}{Theorem}[section]
\newtheorem{lemma}[theorem]{Lemma}
\newtheorem{corollary}[theorem]{Corollary}
\newtheorem{proposition}[theorem]{Proposition}
\newtheorem{fact}[theorem]{Fact}
\theoremstyle{definition}
\newtheorem{problem}[theorem]{Problem}
\theoremstyle{plain}
\numberwithin{equation}{section}

\newcommand{\trans}{\mathsf{T}}
\newcommand{\sfR}{{\mathsf{R}}}
\newcommand{\K}{{\mathsf{K}}}
\newcommand{\E}{{\mathsf{E}}}

\newcommand{\Tau}{{\mathcal T}}

\renewcommand{\L}{\mathsf L}
\newcommand{\G}{{G}}
\newcommand{\FF}{{\mathbb F}}
\newcommand{\Fq}{\mathbb F_q}
\newcommand{\Fqn}{\mathbb F_{q^n}}
\newcommand{\softO}{{O\,\tilde{}\,}}

\newcommand{\Fp}{\mathbb F_p}
\newcommand{\ZZ}{{\mathbb Z}}
\newcommand{\W}{{\mathsf{W}}}
\newcommand{\A}{\mathcal A}

\newcommand{\Moore}{{\mathcal M}}
\newcommand{\nxn}{{n\times n}}

\DeclareMathOperator{\Gal}{Gal}
\DeclareMathOperator{\Tr}{Tr}
\DeclareMathOperator{\Res}{Res}
\DeclareMathOperator{\Circ}{Circ}

\hypersetup{
  colorlinks=true,
  linkcolor=blue!50!black,
  citecolor=blue!50!black,
  urlcolor=blue!60!black,
  pdftitle={Fast Deterministic Normal Bases and Circulant Polynomial Determinants},
  pdfauthor={Mark Giesbrecht, Armin Jamshidpey, \'Eric Schost},
  pdfkeywords={Normal Basis, Finite Fields, Deterministic Algorithm,
    Circulant Determinant, Modular Composition}
}

\title[Fast Deterministic Normal Bases and Circulant Determinants]
  {Fast Deterministic Normal Bases and Circulant Polynomial Determinants}

\thanks{This paper will appear in the proceedings of the 37th International
  Symposium on Algorithms and Computation (ISAAC 2026), December 6--9, 2026}

\author[Mark Giesbrecht]{Mark Giesbrecht}
\address{University of Waterloo, Canada}
\email{mwg@uwaterloo.ca}
\urladdr{https://orcid.org/0009-0006-9312-8768}

\author[Armin Jamshidpey]{Armin Jamshidpey}
\address{University of Waterloo, Canada}
\email{armin.jamshidpey@uwaterloo.ca}
\urladdr{https://orcid.org/0009-0008-3901-9397}

\author[\'Eric Schost]{\'Eric Schost}
\address{University of Waterloo, Canada}
\email{eschost@uwaterloo.ca}
\urladdr{https://orcid.org/0000-0002-8638-9357}

\thanks{Mark Giesbrecht and \'Eric Schost acknowledge support from the
  Natural Sciences and Engineering Research Council of Canada (NSERC)}

\date{}

\begin{document}

\begin{abstract}
  Let $\E=\Fq[x]/(\Gamma)$ describe an algebraic extension of a finite
  field $\Fq$, where $q$ is a prime power and $\Gamma\in\Fq[x]$ is
  monic and irreducible of degree $n$.  We give a deterministic
  algorithm that finds $\beta\in\E$ whose conjugates
  $\beta,\beta^q,\ldots,\beta^{q^{n-1}}$ form an $\Fq$-basis of
  $\E/\Fq$, a \emph{normal basis}, using
  $O_\epsilon((n^2\log q)^{1+\epsilon})+\softO(n\log^2 q)$ bit
  operations for any $\epsilon>0$.

  For $n>1$, let $\theta=x\bmod\Gamma$, so $\E=\Fq[\theta]$. A variant
  of a construction of Artin shows that $\beta_t=(\theta-t)^{-1}$ is
  normal for all but at most $n(n-1)$ parameters $t\in\Fq$.  We
  present an algorithm to construct an $n\times n$ circulant matrix
  over $\Fq[\Tau]$, for an indeterminate $\Tau$, whose determinant at
  $\Tau=t$ is non-zero precisely when $\beta_t$ is normal, and show
  this algorithm requires $\softO(n^2+n\log q)$ operations in $\Fq$.

  Then, as a primary subroutine, using triangular-set power projection
  and modular composition, we show how to compute the determinant of
  any $n\times n$ circulant over $\Fq[\Tau]$, given by its first row
  of polynomials of degree at most $m\geq1$, using
  $O_\epsilon((nm\log q)^{1+\epsilon})$ bit operations. For
  $q\leq n(n-1)$, we show how to embed our problem into a sufficiently
  large field extension, construct a normal basis there, and descend
  to the ground field, within the same stated cost.
\end{abstract}

\maketitle

\section{Introduction}
\label{sec:intro}

Let $\Fq$ be the finite field with $q$ elements.  A finite extension
field $\E=\Fq[x]/(\Gamma)\cong\Fqn$ is constructed via a monic
irreducible polynomial $\Gamma\in\Fq[x]$ of degree $n$. A \emph{normal
  basis} of $\E/\Fq$ is a $\Fq$-basis of the form
\[
  \{\beta,\beta^q,\ldots,\beta^{q^{n-1}}\} \subseteq \E
\]
for some \emph{normal element} $\beta\in\E$.  It is an important
classical result that a normal element $\beta$ always exists.  Our
goal here is to construct such an element and basis deterministically
in nearly quadratic time (i.e., up to poly-logarithmic factors) with
respect to $n$, which is nearly optimal in the output size.

Normal bases have a number of computational advantages over other
bases, such as the power basis for finite field arithmetic. The most
fundamental is that the $q^{th}$-power Frobenius automorphism acts as a
cyclic shift on the coordinates of any element represented in a normal
basis, making the exponentiation by power of $q$ essentially
free. This makes normal bases attractive in cryptographic applications
over binary fields and hardware multipliers for elliptic curve
cryptography~\cite{VerMar21}. The same Frobenius shift property makes
them useful in algorithms for factoring polynomials over finite
fields~\cite{GatSho92, KalSho98}. Normal bases are also a key
ingredient in algorithms for computing isomorphisms and embeddings
between finite fields. The deterministic polynomial-time
algorithm in~\cite{Len91} constructs a normal basis as an essential
subroutine, and subsequent work~\cite{BriDef19} has continued to rely
on this approach. 

The deterministic construction of normal bases has a long
history~\cite{Lun86,Len91,BacDri93}. Augot \& Camion\,\cite{AugCam94}
give a deterministic linear algebra construction requiring
$O(n^3+n^2\log q)$ operations in $\Fq$.  Applying deterministic
Frobenius form algorithms \cite[Corollary~5.2]{NeiPer24} to the
Frobenius map gives an algorithm requiring
$\softO(n^\omega+n^2\log q)$ operations in $\Fq$, where $\omega>2$ is
a feasible exponent for matrix multiplication.

Randomized algorithms are faster in many settings. The Las Vegas
algorithm of~\cite{GatGie90} constructs normal elements, by sampling
and fast testing, with $\softO(n^2+n\log q)$ operations in $\Fq$. The
fast Frobenius and modular-composition techniques of \cite{KalSho98},
together with later fast modular-composition work such as
\cite{KedUma11}, provide important arithmetic primitives for this
current paper, and give a sub-quadratic cost to construct a single
normal element.  For abelian and metacyclic extensions in
characteristic zero, \cite{GieJam19,GieJam21} give Las Vegas
constructions with sub-quadratic complexity.

A basic distinction is between constructing one normal element and
explicitly writing the whole normal basis.  A normal element in the
input power basis has output size $n$, while the full list of
conjugates, or a conversion matrix between the power basis and the
normal basis, has output size $n^2$. Thus, quadratic time is an
output-size barrier for full bases.  Note that conversion between the
power basis and a normal basis can be done in randomized subquadratic
time~\cite{KalSho98,KedUma11}. Once a normal element is known, the
full normal basis can be written out deterministically in essentially
quadratic time as well, as its conjugates form a Frobenius orbit
computable by the iterated Frobenius method of von zur Gathen \&
Shoup\,\cite{GatSho92}.

Although an fully algebraic cost model, counting $\Fq$-operations at
unit cost, would be preferable, our complexity goals appear out of
reach there.  Some core operations rely on the modular composition
algorithm of Kedlaya \& Umans\,\cite{KedUma11} and employ a boolean
RAM, so we state our final costs in bit operations.

Our deterministic algorithm finds a normal element using
\[
  \label{eq:complexity}
  O_\epsilon((n^2\log q)^{1+\epsilon})+\softO(n\log^2 q)
\]
bit operations, for any fixed $\epsilon>0$, considered constant in the
$O_\epsilon(\cdot)$ estimate.  See Theorem~\ref{thm:main}.  Here
$\softO$ suppresses polylogarithmic factors.  Our main subroutine in
Section~\ref{sec:deterministic-circulants} computes the determinant of
an $n\times n$ circulant over $\Fq[\Tau]$, given by its first row of
polynomials of degree at most $m\geq1$, in
$O_\epsilon((nm\log q)^{1+\epsilon})$ bit operations.

We assume henceforth that $n>1$.  The algorithm proceeds as follows.
Let $\theta = x\bmod \Gamma$ be the image of $x$ in $\E$, so
$\E=\Fq[\theta]$.  For any $t\in\Fq$ define
\[
  \beta_t=(\theta-t)^{-1}\in\E.
\]
If we write the conjugates of $\theta$ as
$\theta_i=\theta^{q^i}\in\E$, then the conjugates of $\beta_t$ have
the form $(\theta_i-t)^{-1}$.  The Moore determinant criterion
\cite{Moo96} says that $\beta_t$ is normal if and only if the
determinant of its Moore matrix (whose $i$th row consists of the
successive conjugates of $(\theta_i-t)^{-1}$) is nonzero.  After a
suitable row permutation and clearing some denominators, the Moore
matrix becomes a circulant matrix over $\E$ with entries that depend
polynomially on $t$. Viewing the determinant of this matrix as a
polynomial in $\E[\Tau]$, for a new indeterminate $\Tau$, we prove
that all choices of $\beta_t$ are normal elements, except at most
$n(n-1)$.

The family $\beta_t = (\theta - t)^{-1}$ is a variant of Artin's
construction \cite[Chapter~II, \S~N, Theorem~28]{Art44}.  Gao
\cite[Theorem~3.1.1]{Gao93} develops this in the context of finite
fields noting the polynomial
$\lambda(\Tau)=\Gamma(\Tau)/((\Tau-\theta)\Gamma'(\theta))$ has normal
values at all but at most $n(n-1)$ parameters in $\Fq$.

A Moore matrix with entries in $\E[\Tau]$ is too costly to work with
directly. Instead, we use the equivalent \emph{trace Gram matrix}.
This is a circulant matrix with entries
$g_0(\Tau),\ldots,g_{n-1}(\Tau)$ that are polynomials over our ground
field $\Fq$, and normality of $\beta_t$ is detected by the polynomial
\[
  D(\Tau)=\det\Circ(g_0(\Tau),\ldots,g_{n-1}(\Tau))\in\Fq[\Tau].
\]
We show that $D$ is nonzero, and that for $t\in\Fq$, $D(t) \ne 0$ if
and only if $\beta_t$ is a normal element; in particular, $D$ has at
most $n(n-1)$ roots in $\FF_q$. The polynomials $g_i$ have degree at
most $2n-2$, so $D$ itself has degree at most $2n(n-1)$.  When
$q>n(n-1)$ this gives us a way of constructing a normal element once
$D(\Tau)$ is known by evaluating it at $n(n-1)+1$ elements in $\FF_q$.

The polynomial $D(\Tau)$ can be seen as the cyclic resultant
$\Res_z(z^n-1,A(z,\Tau))\in\Fq[\Tau]$ for a polynomial $A$ of degree
less than $n$ in $z$ and at most $2n-2$ in $\Tau$, built from
$g_0,\dots,g_{n-1}$. The second part of this paper gives a
deterministic algorithm for computing such determinants with near
linear time in the input/output size (in general $O(nm)$), with
$\deg_\Tau A\leq m$ and $m\geq 1$.

The paper is organized as follows. Section~\ref{sec:prelim} fixes
notation and recalls the trace, Moore, and circulant criteria we rely
upon. Section~\ref{sec:criterion} establishes the circulant criterion
for normality: it constructs the trace Gram circulant, bounds the
number of bad parameters by $n(n-1)$, and reduces the construction of
a normal element to a single circulant determinant.
Section~\ref{sec:deterministic-circulants} gives the deterministic
near-linear-time algorithm for this polynomial circulant matrix
determinant. Finally, Section~\ref{sec:small-fields} collects the
routines together and removes the requirement that $q>n(n-1)$ by
passing to a small extension, constructing a normal element there, and
descending it deterministically to $\E/\Fq$.  Theorem~\ref{thm:main}
establishes the final complexity bound.

%%% Local Variables:
%%% mode: LaTeX
%%% TeX-master: "detnorm"
%%% End:

\section{Preliminaries}
\label{sec:prelim}
We recall the notation of the introduction: $\E=\Fq[x]/(\Gamma)$ with
$\Gamma\in\Fq[x]$ monic irreducible of degree $n>1$, and $\theta$ the
image of $x$ in $\E$. Let $ \sigma(a)=a^q$ be the Frobenius
automorphism of $\E/\Fq$. Since finite fields are perfect, the
extension $\E/\Fq$ is separable. Its trace is
$\Tr_{\E/\Fq}(a)=\sum_{i=0}^{n-1} \sigma^i(a),$ and the trace pairing
$(a,b)\longmapsto \Tr_{\E/\Fq}(ab)$ is nondegenerate.  We use the
standard trace-discriminant criterion: elements $a_1,\ldots,a_n\in\E$
form an $\Fq$-basis of $\E$ if and only if
\[
  \det\bigl(\Tr_{\E/\Fq}(a_i a_j)\bigr)_{1\le i,j\le n}\neq 0;
\]
see \cite[Theorem~2.37]{LidNie97}. 
We also use the classical determinant criterion \cite{Moo96}.  For
elements $a_0,\ldots,a_{n-1}$ in a finite extension of $\Fq$, define
the Moore matrix
\[
  \Moore(a_0,\ldots,a_{n-1}) =\bigl(a_j^{q^i}\bigr)_{0\le i,j<n}.
\]
Then
\[
  \det \Moore(a_0,\ldots,a_{n-1})\neq 0 \quad\Longleftrightarrow\quad
  a_0,\ldots,a_{n-1}\text{ are linearly independent over }\Fq.
\]
In particular, $\beta\in\E$ is normal over $\Fq$ if and only if the
Moore determinant of $\beta,\beta^q,\ldots,\beta^{q^{n-1}}$
is nonzero.

For a commutative ring $\W$ and elements $a_0,\ldots,a_{n-1}\in\W$, we
write
\[
  \Circ(a_0,\ldots,a_{n-1}) =
  \begin{pmatrix}
    a_0           & \cdots & a_{n-1} \\
    % a_{n-1}      & \cdots & a_{n-2} \\
    \vdots    & \ddots & \vdots \\
    a_1           & \cdots & a_0
  \end{pmatrix}
  =
  \bigl(a_{s-r\bmod n}\bigr)_{0\le r,s<n}
  \in\W^\nxn
\]
for the corresponding circulant matrix. If
$A(z)=\sum_{i=0}^{n-1}a_i z^i$, then
$\det\Circ(a_0,\ldots,a_{n-1})=\Res_z(z^n-1,A(z))\in\W$.

We assume fast polynomial arithmetic, so multiplication of polynomials
of degree $d$ over $\FF_q$ costs $\softO(d)$ operations in $\Fq$; as a
consequence, arithmetic operations in $\E$ take $\softO(n)$ operations
in $\Fq$.  The construction below requires all the Frobenius
conjugates of the power-basis generator $\theta$,
\[
  \theta_i=\theta^{q^i},\qquad 0\le i<n.
\]
Binary powering followed by the iterated Frobenius algorithm
\cite[Algorithm~3.1 and Theorem~3.2]{GatSho92} computes
$\theta_0,\ldots,\theta_{n-1}$ in
$\softO(n^2+n\log q)$ operations in $\Fq$.

%%% Local Variables:
%%% mode: LaTeX
%%% TeX-master: "detnorm"
%%% End:

\section{A circulant criterion for normality}
\label{sec:criterion}

In this section, we first show how to construct a normal element
$\beta_t$ for $\E/\Fq$ from an element $t\in\Fq$ that does not cancel
the determinant of a certain Moore circulant matrix $B(\Tau)$ over
$E[\Tau]$.  We then introduce a circulant matrix $\G(\Tau)$ over
$\Fq[\Tau]$, with the property that $\det G=(\det B)^2$.  Our
construction is similar to Artin's\,\cite{Art44}; see Gao
\cite[Theorem~3.1.1]{Gao93}.  For $t\in\Fq$, let
$\beta_t=(\theta-t)^{-1}\in\E$.  This is well-defined since the
irreducible polynomial $\Gamma$ has no root in $\Fq$ so
$\Gamma(t)\neq 0$.

\subsection{The Moore circulant matrix}

Write $\theta_i=\theta^{q^i}$ for $0\le i<n$, with indices taken
modulo $n$.  The roots of $\Gamma$ are exactly the $\theta_i$, and
hence, if we let $\Tau$ be a new indeterminate, we have
\[
  \Gamma(\Tau)=\prod_{i=0}^{n-1}(\Tau-\theta_i) \in \E[\Tau].
\]
For $0\le i<n$, define
\[
  b_i(\Tau)=\frac{\Gamma(\Tau)}{\theta_i-\Tau}
           =-\prod_{j\neq i}(\Tau-\theta_j)\in\E[\Tau].
\]

\begin{lemma}[Cleared Moore circulant]
  \label{lem:moore-circ}
  Let
  \[
    B(\Tau)=\Circ(b_0(\Tau),\ldots,b_{n-1}(\Tau))\in\E[\Tau]^{n\times n}.
  \]
  Then, for every $t\in\Fq$, $\beta_t$ is normal over $\Fq$ if and only
  if $\det B(t)\neq 0$.  Moreover, $\det B(\Tau)$ is a nonzero
  polynomial in $\E[\Tau]$ of degree at most $n(n-1)$.
\end{lemma}

\begin{proof}
  For $t\in\Fq$, the conjugates of $\beta_t$ are
  $\beta_t^{q^i}=(\theta_i-t)^{-1}$ for $0\le i<n$.  After the row 
  permutation $r\mapsto -r$, their Moore matrix is
  \[
    C(t)=\Circ\bigl((\theta_0-t)^{-1},\ldots,
    (\theta_{n-1}-t)^{-1}\bigr).
  \]
  Moore's criterion indicates normality exactly when $\det C(t)\neq 0$.
  Since $\Gamma(t)\neq 0$ and $B(t)=\Gamma(t)C(t)$, this is equivalent
  to $\det B(t)\neq 0$.

  Each $b_i$ has degree $n-1$, so $\deg\det B\le n(n-1)$.  To prove
  that this determinant is not identically zero, evaluate at
  $\Tau=\theta_j$.  Then $b_i(\theta_j)=0$ for $i\neq j$, and
  $b_j(\theta_j)=-\Gamma'(\theta_j)\neq 0$, because $\Gamma$ is
  separable.  Thus $B(\theta_j)$ is a nonzero scalar multiple of a
  cyclic permutation matrix, and so its determinant is nonzero.
\end{proof}

This immediately yields the following corollary.
\begin{corollary}[Bad parameters]
  \label{cor:badset}
  The set of $t\in\Fq$ for which $\beta_t$ is not normal has size at
  most $n(n-1)$.  In particular, if $q>n(n-1)$, then some
  $(\theta-t)^{-1}$ is normal over $\Fq$.
\end{corollary}

\subsection{A ground-field circulant matrix criterion}
\label{sec:gram}

The matrix $B$ proves the bound above, but its entries lie in
$\E[\Tau]$.  The next lemma gives a polynomial analogue of the
classical trace-circulant criterion; see \cite[Theorem~2.2.5]{Gao93} for
our reciprocal family, yielding a circulant matrix over
$\FF_q[\Tau]$. A similar construction is found in Gao\,\cite[proof of
Theorem~3.1.1]{Gao93}.

For $0\le k<n$, define
\[
  g_k(\Tau)=\sum_{i=0}^{n-1} b_i(\Tau)b_{i+k}(\Tau),
\]
with indices modulo $n$, and set
\[
  \G(\Tau)=\Circ(g_0(\Tau),g_1(\Tau),\ldots,g_{n-1}(\Tau)).
\]

\begin{lemma}[Cleared trace Gram circulant]
  \label{lem:gram-circ}
  The polynomials $g_0,\ldots,g_{n-1}$ lie in $\Fq[\Tau]$, and
  \[
    \G(\Tau)=B(\Tau)^{\top}B(\Tau)\in\Fq[\Tau]^{n\times n}.
  \]
  Hence, $\det \G(\Tau)=(\det B(\Tau))^2$, so for every $t\in\Fq$,
  \[
    \det\G(t)\neq 0
    \quad\Longleftrightarrow\quad
    \beta_t\text{ is normal over }\Fq.
  \]
  Moreover, $\deg g_k\le 2n-2$ for all $k$.
\end{lemma}

\begin{proof}
  Extend Frobenius to $\E[\Tau]$ by fixing $\Tau$.  Since
  $\sigma(b_i)=b_{i+1}$, cyclic reindexing gives
  \[
    \sigma(g_k)=\sum_i b_{i+1}b_{i+k+1}=g_k.
  \]
  Thus $g_k\in\Fq[\Tau]$.  The degree bound follows from
  $\deg b_i=n-1$.  By the definition of a circulant matrix,
  \[
    (B^{\trans}B)_{r,\ell} =\sum_{i=0}^{n-1} b_{r-i}b_{\ell-i}
    =\sum_{j=0}^{n-1} b_jb_{j+\ell-r} =g_{\ell-r},
  \]
  so $B^{\trans}B=\G$.  Therefore $\det\G(t)=\det(B(t))^2$, and
  the equivalence with normality follows from
  Lemma~\ref{lem:moore-circ}.
\end{proof}

Note that for $t\in\Fq$, $\Gamma(t)^{-2}\G(t)$ is the trace Gram matrix of the conjugates
of $\beta_t$.  Indeed,
if $0\le r,\ell<n$, then
\[
  \Gamma(t)^{-2}\G(t)_{r,\ell}
  =\sum_{i=0}^{n-1}
     \frac{1}{(\theta_i-t)(\theta_{i+\ell-r}-t)}
  =\Tr_{\E/\Fq}\bigl(\beta_t^{q^r}\beta_t^{q^\ell}\bigr).
\]

\subsection{Computing the first row of $\G$}

We now compute the first row of $\G$ efficiently by the same product
calculation that underlies Lagrange interpolation at the conjugates of
$\theta$.  The polynomial
\[
  \frac{\Gamma(\Tau)}{\Tau-\theta_i}\in\E[\Tau]
\]
vanishes at all conjugates of $\theta$ except $\theta_i$, where its
value is $\Gamma'(\theta_i)$.  The cyclic sums used below are fixed by
Frobenius and therefore lie in $\Fq[\Tau]$.  Thus the needed values are
obtained from
\[
  \Gamma'(\theta)=\prod_{j=1}^{n-1}(\theta-\theta_j)
\]
by deleting one factor.

\begin{lemma}[First row formula]
  \label{lem:firstrow}
  Let $\Gamma'$, $\Gamma''\in\Fq[\Tau]$ denote the first and second formal
  derivative of $\Gamma$.  Then
  \[
    g_0(\Tau)=\Gamma'(\Tau)^2-\Gamma(\Tau)\Gamma''(\Tau)\in\Fq[\Tau].
  \]
  For $1\le j<n$, put
  \[
    d_j=\theta-\theta_j,
    \qquad
    h_j=\prod_{\substack{1\le m<n\\ m\neq j}} d_m\in\E.
  \]
  Empty products are understood to be $1$.  For $1\le k<n$, let
  $H_k(x)\in\Fq[x]$ be the representative of $h_k+h_{n-k}$ of degree
  less than $n$.  Then
  \[
    g_k(\Tau)=\Gamma(\Tau)H_k(\Tau).
  \]
  In particular, $\deg H_k\le n-2$ and $g_{n-k}=g_k$.
\end{lemma}

\begin{proof}
  Since
  \[
    g_0(\Tau)=\Gamma(\Tau)^2
       \sum_{i=0}^{n-1}\frac{1}{(\Tau-\theta_i)^2},
  \]
  the logarithmic derivative identity
  \[
    \frac{\Gamma'(\Tau)}{\Gamma(\Tau)}
      =\sum_{i=0}^{n-1}\frac{1}{\Tau-\theta_i}
  \]
  gives, after differentiation,
  \[
    \sum_i\frac{1}{(\Tau-\theta_i)^2}
      =\frac{\Gamma'(\Tau)^2-\Gamma(\Tau)\Gamma''(\Tau)}
             {\Gamma(\Tau)^2}.
  \]
  This proves the formula for $g_0$.

  Now let $1\le k<n$.  From the definition of $g_k$,
  \[
    g_k(\Tau)=\Gamma(\Tau)Q_k(\Tau), \quad \mbox{where}\quad
    Q_k(\Tau)=\sum_{i=0}^{n-1}
    \frac{\Gamma(\Tau)}{(\Tau-\theta_i)(\Tau-\theta_{i+k})}.
  \]
  Each summand is a polynomial of degree $n-2$, and the sum is fixed by
  Frobenius, so $Q_k\in\Fq[\Tau]$ and $\deg Q_k\le n-2$.

  Evaluate $Q_k$ at $\theta=\theta_0$.  All summands vanish except
  those with $i=0$ and $i=n-k$.  Hence
  \[
    Q_k(\theta)
      =\frac{\Gamma'(\theta)}{\theta-\theta_k}
       +\frac{\Gamma'(\theta)}{\theta-\theta_{n-k}}.
  \]
  Since
  \[
    \Gamma'(\theta)=\prod_{j=1}^{n-1}(\theta-\theta_j)
                   =\prod_{j=1}^{n-1}d_j,
  \]
  this gives
  \[
    Q_k(\theta)=h_k+h_{n-k}.
  \]
  Evaluation at $\theta$ is an $\Fq$-linear isomorphism from the from
  the polynomials in $\Fq[\Tau]$ of degree less than $n$ to $\E$, so
  $Q_k=H_k$.  Re-indexing the sum gives $g_{n-k}=g_k$.
\end{proof}

\begin{proposition}[Cost of building $\G$]
  \label{prop:buildG}
  The polynomials $g_0,\ldots,g_{n-1}$ can be computed
  in $\softO(n^2+n\log q)$ operations in $\Fq$.  They
  satisfy $\deg g_k\le 2n-2$ for all~$k$.
\end{proposition}

\begin{proof}
  Compute the Frobenius table $\theta_0,\theta_1,\ldots,\theta_{n-1}$
  by the method recalled in Section~\ref{sec:prelim}.  This costs
  $\softO(n^2+n\log q)$ operations over $\Fq$.  Form
  $d_j=\theta-\theta_j$ for $1\le j<n$.  Then compute all products
  \[
    h_j=\prod_{\substack{1\le m<n\\m\neq j}}d_m
  \]
  by one forward sweep and one backward sweep.  Set
  $  L_0=1,
    \qquad
    L_j=d_1d_2\cdots d_j\quad(1\le j<n).$
  Then sweep backwards with $S=1$; at step $j=n-1,n-2,\ldots,1$, set
  \[
    h_j=L_{j-1}S,
    \qquad
    S\leftarrow Sd_j.
  \]
  This uses $O(n)$ multiplications in $\E$, hence $\softO(n^2)$
  operations over $\Fq$.  Compute
    $g_0=\Gamma'^2-\Gamma\Gamma''.$
  
  For $1\le k\le \lfloor n/2\rfloor$, let $H_k$ be the representative
  of $h_k+h_{n-k}$ of degree less than $n$, and set
  \[
    g_k=\Gamma H_k,
    \qquad
    g_{n-k}=g_k.
  \]
  If $k=n-k$, this assignment is made only once.  The products
  $\Gamma H_k$ cost $\softO(n^2)$ operations in total.  The formulas
  and degree bounds are those of Lemma~\ref{lem:firstrow}.
\end{proof}

Set
\[
  D(\Tau)=\det\G(\Tau)
  =\det\Circ(g_0(\Tau),\ldots,g_{n-1}(\Tau))\in\Fq[\Tau].
\]
For every $t\in\Fq$, $D(t)\neq 0$ if and only if $\beta_t$ is normal
over $\Fq$.  Although $\deg D\leq2n(n-1)$, the identity $D=(\det B)^2$
bounds its number of zeros by $n(n-1)$ from Corollary~\ref{cor:badset}.

%%% Local Variables:
%%% mode: LaTeX
%%% TeX-master: "detnorm"
%%% End:

\section{Computing determinants of circulant polynomial matrices}
\label{sec:deterministic-circulants}

We now isolate the remaining determinant computation.

\begin{problem}[Circulant determinant]
  \label{prob:circulant-det}
  Given polynomials $a_0(\Tau),\ldots,a_{n-1}(\Tau)\in\Fq[\Tau]$ of
  degree at most $m$, for some $m > 0$, compute
  \[
    \det C(\Tau)~~\mbox{for}~~
    C(\Tau)=\Circ\left(a_0(\Tau),a_1(\Tau),\ldots,a_{n-1}(\Tau)\right)
    \in\Fq[\Tau]^{n\times n}.
  \]
  Equivalently, for $ A(z,\Tau)=\sum_{i=0}^{n-1}a_i(\Tau)z^i$,
  compute $ \Res_z(z^n-1,A(z,\Tau))$.
\end{problem}

Let $\Phi(z)=z^n-1$. We prove the following theorem in the rest of the
section.
\begin{theorem}[polynomial circulant determinant]
  \label{thm:circulant-main}
  Suppose that $\FF_q$ is presented as $\FF_p[y]/(\Lambda(y))$, for some
  irreducible polynomial $\Lambda \in \FF_p[y]$. Then the polynomial
  \[
    \det C(\Tau)=\Res_z(\Phi(z),A(z,\Tau))\in\FF_q[\Tau]
  \]
  can be computed using $O_\epsilon((nm\log q)^{1+\epsilon})$ bit
  operations, for any $\epsilon > 0$.
\end{theorem}
For the trace-Gram circulants constructed in
Proposition~\ref{prop:buildG}, one has $m\le 2n-2$, and hence $nm \in
O(n^2)$. 

If $A$ is monic in $\Tau$, the problem reduces to computing the
characteristic polynomial of multiplication by $\Tau$ in the
$\FF_q$-algebra $\A=\FF_q[z,\Tau]/(\Phi(z), A(z,\Tau))$, and $A$ being
monic allows us to use existing fast algorithms for triangular sets to
conclude.

There is however no guarantee that $A$ is monic in $\Tau$. The so called
{\em dynamic evaluation} techniques allow us to reduce to this
situation, but it requires $\Phi$ be squarefree, which we don't
assume (this is equivalent to $\gcd(n,p)=1$). We handle this issue in
Subsection~\ref{ssec:sqfree}: due to its special shape, computing the
squarefree decomposition of $\Phi$ is straightforward, and the
multiplicativity of the resultant allows us to work with the
squarefree part $\varphi$ of $\Phi$ instead of $\Phi$ itself.

In Subsection~\ref{ssec:monic}, we reduce to the case where $A$ is
monic. The obvious idea is to divide $A$ by its leading coefficient
$\ell$ in $\Tau$, working modulo $\varphi$. However, we cannot assume
that $\varphi$ is irreducible, so $\ell$ may have a non-trivial GCD
with $\varphi$. This gives us a factorization of $\varphi$ into
$\gcd(\varphi,\ell)$, modulo which $\ell$ vanishes, and
$\varphi/\gcd(\varphi,\ell)$, modulo which $\ell$ is a unit.
Continuing this process leads us to a factorization of $\varphi$ into
coprime factors $H_1,\dots,H_s$, modulo which we can guarantee that
the leading coefficient of $A$ is invertible.

Subsection~\ref{ssec:fast-tri} handles the core question of computing
$\Res_z(H(z),P(z,\Tau))$, with $H$ and $P$ monic in respectively $z$
and $\Tau$. We observe that this is equivalent to computing the
characteristic polynomial of the endomorphism of multiplication by
$\Tau$ modulo $(H(z),P(z,\Tau))$, which we proceed to handle by a
version of Leverrier's algorithm. 

%%%%%%%%%%%%%%%%%%%%%%%%%%%%%%%%%%%%%%%%%%%%%%%%%%%%%%%%%%%%

\subsection{Reduction to the squarefree case}\label{ssec:sqfree}

The first step in the whole procedure is to compute the squarefree
decomposition of $\Phi$. Write $n=p^c r$, with $\gcd(r,p)=1$; it
follows that in $\FF_q[z]$, we have $\Phi = z^n-1 = \varphi^{p^c}$,
with $\varphi=z^r-1$ squarefree. The integers $c$ and $r$ are computed
on a Boolean RAM as follows: given the base-2 expansion of $n$, we can
obtain its base-$p$ expansion in $\softO(\log n)$ bit
operations~\cite[Theorem~9.17]{GatGer13}. From this, we deduce $c$ and
$r$ (in base 2) in softly linear time again.

Let $B(z,\Tau) = A(z,\Tau) \bmod \varphi$; we can compute it using
$\softO(n m)$ operations in $\FF_q$ by Euclidean division
coefficientwise in $\Tau$. By multiplicativity of the resultant, we
deduce that
\[ \Res_z(\Phi(z),A(z,\Tau)) = \Res_z(\varphi(z),B(z,\Tau))^{p^c},\]
so we can focus on computing $\Res_z(\varphi(z),B(z,\Tau))$. This is a
polynomial of degree at most $r m$ in $\FF_q[\Tau]$, and raising it to
$p^c$-th power to obtain $\Res_z(\Phi(z),A(z,\Tau))$ also takes $\softO(n
m)$ operations in $\FF_q$ by repeated squaring.

%%%%%%%%%%%%%%%%%%%%%%%%%%%%%%%%%%%%%%%%%%%%%%%%%%%%%%%%%%%%

\subsection{Reduction to monic inputs}\label{ssec:monic}

In order to compute $\Res_z(\varphi(z),B(z,\Tau))$, the second step is
to reduce to the case where $B(z,\Tau)$ is monic in $\Tau$. If
$\varphi$ were irreducible, this would amount to inverting the leading
coefficient of $B$ modulo $\varphi$, but all we know is that $\varphi$
is squarefree. Instead, we use the notion of {\em monic form}
from~\cite[Definition~4.2]{DahMor06}: a monic form for the pair
$(\varphi,B)$ is a sequence $((H_1,B_1,c_1),\dots,(H_s,B_s,c_s))$,
that satisfies the following:
\begin{itemize}
\item for all $i$, $H_i$ and $c_i$ are in $\FF_q[z]$, $H_i$ is monic
  of positive degree and $c_i$ is reduced with respect to $H_i$
\item $\varphi=H_1 \cdots H_s$
\item for all $i$, $B_i = B \bmod H_i$ in $\FF_q[z,\Tau]$
\item for all $i$, either $B_i=c_i=0$ or $P_i=c_i B_i \bmod H_i$ is
  monic in $\Tau$.
\end{itemize}
Using Proposition~2.4 in~\cite{DahMor06} and Algorithm~{\tt monic}
from that reference, a monic form can be computed in~$\softO(r m)$
operations in $\FF_q$.

Using again the multiplicativity of the resultant, we deduce that
\[ \Res_z(\varphi(z),B(z,\Tau)) = \prod_{i \le s} \Res_z(H_i(z),B(z,\Tau)) = \prod_{i \le s} \Res_z(H_i(z),B_i(z,\Tau)).\]
If any $B_i$ vanishes, this implies that $\Res_z(\varphi(z),B(z,\Tau))=0$, so we are done.
Else, we obtain
\[ \Res_z(\varphi(z),B(z,\Tau)) = \prod_{i \le s} \Res_z(H_i(z), b_i(z) P_i(z,\Tau)) =R_1 R_2\]
with $b_i$ the leading coefficient of $B_i$ in $\Tau$ and where
\[ R_1 = \prod_{i \le s} \Res_z(H_i(z), b_i(z)) \quad\text{and}\quad
R_2 = \prod_{i \le s} \Res_z(H_i(z), P_i(z,\Tau))\] Computing each scalar
resultant $\Res_z(H_i(z), b_i(z))$ is done in quasi-linear time in the
degree of the corresponding $H_i$, so it remains to compute each
$\Res_z(H_i(z), P_i(z,\Tau))$. If we write $d_i = \deg_z H_i$ and $e_i
=\deg_\Tau P_i$, the following subsection establishes that it can be done
in $O_\epsilon((d_ie_i\log q)^{1+\epsilon})$ bit operations for any
$\epsilon > 0$.

%%%%%%%%%%%%%%%%%%%%%%%%%%%%%%%%%%%%%%%%%%%%%%%%%%%%%%%%%%%%

\subsection{Fast triangular-set machinery}
\label{ssec:fast-tri}

Let $\sfR$ be a commutative ring and set
\[
  \A=\sfR[z,\Tau]/(H(z),P(z,\Tau)),
\]
where $H\in\sfR[z]$ is monic of degree $d$ and $P\in \sfR[z,\Tau]$ is
monic of degree $e\geq 1$ in $\Tau$, and reduced with respect to $H$. Then $\A$
is a free $\sfR$-module of rank $de$, with monomial basis $\mathcal B=(z^i
\Tau^j)_{0\le i < d, 0\le j< e}$. To any element $\alpha$ in $\A$, one
associates the endomorphism $\mu_\alpha:\A\to\A$ of multiplication by
$\alpha$; the trace $\Tr_{\A/\sfR}(\alpha) \in \sfR$ and the
characteristic polynomial $\chi_{\A/\sfR}(\alpha)\in \sfR[\Tau]$ of
$\alpha$ are by definition those of $\mu_\alpha$. By extension, the
trace, resp.\ characteristic polynomial of $a \in \sfR[z,\Tau]$ are
defined as those of its residue class in $\A$.

The following lemma is folklore; we will only need it over a field,
but establishing it over an arbitrary ring takes hardly more work.
\begin{lemma}
  $\chi_{\A/\sfR}(\Tau) = \Res_z(H(z),P(z,\Tau))$.
\end{lemma}
\begin{proof}
  Write $H=z^d + \sum_{i < d} h_i z^i$ and
  $P=\Tau^e + \sum_{i < d, j<e} p_{i,j} z^i\Tau^j$; let further
  $(\bar h_i)_{i <d}$ and $(\bar p_{i,j})_{i < d,j<e}$ be new
  indeterminates over $\mathbb Q$, which we use as coefficients to
  define polynomials $\bar H$ and $\bar P$.  Both the characteristic
  polynomial $C$ of $\Tau$ modulo $(\bar H,\bar P)$ and the resultant
  $R=\Res_z(\bar H,\bar P)$ are polynomials in
  $\ZZ[\Tau,\bar h_i,\bar p_{i,j}]$, and $\chi_{\A/\sfR}(\Tau)$ and
  $\Res_z(H,P)$ are obtained by evaluating $C$ and $R$ at the
  coefficients $h_i$ and $p_{i,j}$ in $\sfR$. This is by construction
  for the characteristic polynomial, and follows from $\bar H$ being
  monic in $z$ for the resultant, as in~\cite[Lemma~6.25]{GatGer13}.

  This means that it suffices to prove the claim for the polynomials
  $\bar H$ and $\bar P$. The former splits as a product of $d$
  distinct linear factors $\prod_{i \le d}(z-z_i)$ over an algebraic
  closure of its field of definition. By the Chinese Remainder theorem
  for the left-hand side, and multiplicativity of the resultant for
  the right-hand side, it suffices to consider the case where
  $H=(z-z_i)$, in which case both sides of the equality are seen to be
  $P(z_i,\Tau)$.
\end{proof}

This lemma reduces the computation of the resultant
$\Res_z(H(z),P(z,\Tau))$ to that of the characteristic polynomial
$\chi_{\A/\sfR}(\Tau)$. This will be done by means of Leverrier's
algorithm, which relies on the Newton identities
\[
  j c_j=-\sum_{i=1}^j c_{j-i}s_i, \qquad 1\le j\le de,
\]
where we write $c_0=1$, $s_k=\Tr_{\A/\sfR}(\Tau^k)$ for $ k \ge 0$ and
$\chi_{\A/\sfR}(\Tau) =\Tau^{de}+c_1\Tau^{de-1}+\cdots+c_{de}.$

Given the traces $s_0,\dots,s_{de}$, using the recurrence above to
compute the $c_j$'s results in a number of operations in $\sfR$ that
is quadratic in ${de}$; a more efficient approach based on Newton
iteration runs in quasi-linear time. However, both algorithms require
divisions by integers up to ${de}$ in $\sfR$, and these integers may
vanish if the characteristic of $\sfR$ is too small.

When we work over a finite field, though, this issue can be
circumvented by working with (truncated) $p$-adic integers rather than
finite field elements. Thus, let us from now on assume that
$\sfR=\FF_q$, presented as $\FF_q=\FF_p[y]/(\Lambda(y))$, for some
monic irreducible polynomial $\Lambda \in \FF_p[y]$ of degree $\kappa$. For
$K \ge 1$, define further $\sfR_K=\ZZ[y]/(p^K,\tilde\Lambda(y))$,
where $\tilde\Lambda$ is an arbitrary monic lift of $\Lambda$ to
$\ZZ[y]$ of degree $\kappa$.

The following fact is proved in~\cite{BosGon05} when $\kappa=1$ (so in
that case $\FF_q=\FF_p$), but the proof holds for an arbitrary
extension degree. It states that if we work with $p$-adic lifts of our
input polynomials, a logarithmic precision (in ${de}$) is sufficient to allow us
to compute the characteristic polynomial of $\Tau$ in $\A$.  Because
this involves calculations over $\sfR_K$ rather than $\sfR=\FF_q$, the
arithmetic complexity model is not appropriate to estimate the runtime
of this operation. Instead, we measure its cost in terms of bit
operations on a RAM.
\begin{fact}\label{fact:charpoly}
  Let $K=1+\lceil \log_p({de})\rceil$, and let $\tilde H$ and $\tilde P$ be
  arbitrary lifts of $H$ and $P$ to respectively $\sfR_K[z]$ and
  $\sfR_K[z,\Tau]$, with still $\tilde H$ monic of degree $d$ and $\tilde
  P$ monic of degree $e$ in $\Tau$ and reduced with respect to $\tilde H$.
  Put $\A_K=\sfR_K[z,\Tau]/(\tilde H(z),\tilde P(z,\Tau))$.

  Given the traces
  $\Tr_{\A_K/\sfR_K}(1),\Tr_{\A_K/\sfR_K}(\Tau),\ldots,\Tr_{\A_K/\sfR_K}(\Tau^{de})$,
  one can compute $\chi_{\A/\sfR}(\Tau)$ using $ \softO\bigl({de} \log q)$
  bit operations.
\end{fact}

Fix $K=1+\lceil \log_p({de})\rceil$ as above.  Given an $\sfR_K$-linear
form $\tau:\A_K\to\sfR_K$ by means of its values on the monomial
basis $\mathcal B$ of $\A_K$, {\em power projection} asks for the sequence
$\tau(1),\tau(\Tau),\ldots,\tau(\Tau^{de}).$

The reconstruction algorithm takes as input these values for $\tau
=\Tr_{\A_K/\sfR_K}$. The trace values
$\Tr_{\A_K/\sfR_K}$ on $\mathcal B$ can be computed in a quasi-linear
$\softO({de})$ operations in $\sfR_K$,
using~\cite[Proposition~8]{PasSch06}.

Once the traces of all basis elements are known, we use the
two-variable triangular-set power projection algorithm of
\cite[Theorem~3.4]{PotSch13}, which builds on the modular composition
and power projection methods of \cite{KedUma11}. These algorithms also
require that we work with a boolean RAM. The algorithm
of~\cite{PotSch13} is written for polynomials with coefficients in
$\sfR=\FF_q$ rather than $\sfR_K$, but as pointed out
in~\cite{KedUma11, LebMeh13} the techniques carry over to polynomials
over~$\sfR_K$.

\begin{fact}[Triangular power projection]
  \label{fact:tripower}
  Let $K$, $\sfR_K$, $\tilde H$ and $\tilde P$ and $\A_K$ be as in
  Fact~\ref{fact:charpoly}.  The sequence
  \[
  \Tr_{\A_K/\sfR_K}(1),\Tr_{\A_K/\sfR_K}(\Tau),\ldots,\Tr_{\A_K/\sfR_K}(\Tau^{de})
  \]
  can be computed in $O_\epsilon(({de}\log q)^{1+\epsilon})$ bit
  operations for any $\epsilon > 0$.
\end{fact}

Combining Facts~\ref{fact:charpoly} and \ref{fact:tripower}, we obtain
the following proposition.

\begin{proposition}[characteristic polynomial / resultant]
  \label{prop:char}
  Let $H\in\FF_q[z]$ be monic of degree $d\geq 1$ and $P\in \FF_q[z,\Tau]$
  be monic of degree $e\geq 1$ in $\Tau$, and reduced with respect to $H$,
  where $\FF_q$ is presented as $\FF_p[y]/(\Lambda(y))$, for some
  monic irreducible polynomial $\Lambda \in \FF_p[y]$.
  
  Then the resultant $\Res_z(H(z),P(z,\Tau))$ can be computed in
  $O_\epsilon((de\log q)^{1+\epsilon})$ bit operations for any
  $\epsilon > 0$.
\end{proposition}

%%%%%%%%%%%%%%%%%%%%%%%%%%%%%%%%%%%%%%%%%%%%%%%%%%%%%%%%%%%%

\subsection{Cyclic determinant algorithm and complexity}
\label{ssec:cycdetalg}

We can now summarize the algorithm in the order of the preceding
sections, in order to compute $\Res_z(\Phi(z),A(z,\Tau))$, with
$n=\deg_z \Phi$, $\deg_{\Tau} A \le m$ for some $m > 0$.

\begin{enumerate}[leftmargin=2em]
\item Write $n=p^c r$, with $\gcd(p,r)=1$, $\varphi=z^r-1$ and $B=A
  \bmod \varphi$. This takes $\softO(\log n)$ bit operations and
  $\softO(n m)$ operations in $\FF_q$.
\item Compute a monic form $((H_1,B_1,c_1),\dots,(H_s,B_s,c_s))$
  of $(\varphi,B)$, using $\softO(rm)$ operations in $\FF_q$
\item If any of the $B_i$'s vanishes, return $0$. 
\item Compute $R_1 = \prod_{i \le s} \Res_z(H_i,b_i)$, where for all
  $i$, $b_i$ is the leading coefficient of $B_i$. This takes
  $\softO(n)$ operations in $\FF_q$ using the half-GCD resultant
  algorithm~\cite[Chapter~11]{GatGer13}.
\item For all $i$, set $P_i = c_i B_i \bmod H_i$ and compute
  $\chi_i = \Res_z(H_i,P_i)$ using Proposition~\ref{prop:char}.  The
  former takes $\softO(d_i e_i)$ operations in $\FF_q$, with
  $d_i = \deg_z H_i$ and $e_i=\deg_\Tau P_i$, the latter
  $O_\epsilon((d_ie_i\log q)^{1+\epsilon})$ bit operations for any
  $\epsilon > 0$.

  Since $e_i \le m$ for all $i$ and $d_1 + \cdots + d_s =r$, the total
  cost is $\softO(rm)$ operations in $\FF_q$ and
  $O_\epsilon((r m\log q)^{1+\epsilon})$ bit operations.
\item Compute $R_2 = \prod_{i \le s} \chi_i$ in $\FF_q[\Tau]$; because
  $R_2$ has degree at most $r m$, this takes $\softO(r m)$
  operations in $\FF_q$, using subproduct tree
  techniques~\cite[Chapter~10]{GatGer13}.
\item Return $(R_1 R_2)^{p^c}$. Because the result has degree at most
  $n m$, this takes $\softO(n m)$ operations in $\FF_q$.
\end{enumerate}
Since any operation in $\FF_q$ can be executed on a RAM using
$\softO(\log q)$ bit operations, the full circulant determinant
computation costs $O_\epsilon((n m\log q)^{1+\epsilon})$ bit
operations, for any $\epsilon > 0$, proving
Theorem~\ref{thm:circulant-main}.

%%% Local Variables:
%%% mode: LaTeX
%%% TeX-master: "detnorm"
%%% End:

\section{Main algorithm: deterministic fast normal bases over any finite field}
\label{sec:small-fields}

From what we have seen so far, we can construct a normal element when
$q>n(n-1)$ as follows.  By Proposition~\ref{prop:buildG},
construct $g_0,\ldots,g_{n-1}$ in $\softO(n^2+n\log q)$ operations
over $\Fq$, which is $\softO(n^2\log q+n\log^2 q)$ bit operations.
Then, compute $D(\Tau)=\det\Circ(g_0(\Tau),\ldots,g_{n-1}(\Tau))$
using the circulant determinant routine, for
$O_\epsilon((n^2 \log q)^{1+\epsilon})$ bit operations. Choose any
$n(n-1)+1$ distinct elements of $\Fq$ and evaluate $D$ at them.  By
Corollary~\ref{cor:badset}, at least one of these values, say $t$,
satisfies $D(t)\neq 0$.  Then $(\theta-t)^{-1}$ is normal over $\Fq$.
The multipoint evaluation uses $\softO(n^2)$ operations in $\Fq$, and
the final inversion in $\E$ is a lower order term. The total in this
case is
\[
  O_\epsilon((n^2 \log q)^{1+\epsilon}) + \softO(n\log^2 q)
\] 
bit operations.

\subsection{A reduction for smaller finite fields}

Now suppose $q\leq n(n-1)$.  We proceed by constructing a normal
element over a sufficiently large ground field, and descend it by a
product of conjugate resolvents.  Start by choosing a prime
$\ell=O(\log n\log\log n)$ with $q^\ell>n(n-1)$ and $\gcd(\ell,n)=1$,
by trial division.  Write $q=p^\kappa$. Shoup's algorithm
\cite[Theorem~4.1]{Sho90} constructs a monic irreducible $h\in\Fq[Y]$
of degree $\ell$. Its cost is $\softO(n)$ bit operations, since
$p^{1/2}=O(n)$ and $\kappa,\ell,\log p$ are polylogarithmic in $n$.
Set $\L=\Fq[Y]/(h)$.

To obtain a prime-field power basis for $\L$, compute a cyclic vector
of its $p$-Frobenius map using \cite[Corollary~5.2]{NeiPer24}.  This
vector generates $\L$ over $\Fp$; its minimal polynomial and both
basis-conversion matrices follow by linear algebra.  Since
$[\L:\Fp]=\kappa\ell$ is polylogarithmic in $n$, these computations
take polylogarithmic bit time.  The polynomial $\Gamma$ remains
irreducible over $\L$ \cite[Lemma~2.3.2]{Gao93}. Thus
$\K=\L[x]/(\Gamma)$ contains $\E$, with the following field diagram:
\[
\begin{tikzcd}[column sep=3.5em, row sep=1.4em]
 & \K\simeq\FF_{q^{n\ell}} & \\
 \L=\FF_{q^\ell} \arrow[ur, hook] &&
 \E=\Fq[x]/(\Gamma) \arrow[ul, hook] \\
 & \Fq \arrow[ul, hook] \arrow[ur, hook] &
\end{tikzcd}
\]

As before, let $\sigma$ be the $q$th-power Frobenius automorphism, and
set $\tau=\sigma^\ell$ and $\psi=\sigma^n$.  Then $\tau$ generates
$\Gal(\K/\L)$, while $\psi$ generates $\Gal(\K/\E)$.  The restriction
of $\tau$ to $\E$ also generates $\Gal(\E/\Fq)$.

Run our normal-element algorithm over the extension $\K/\L$, which
works since the ground field has $|\L|=q^\ell>n(n-1)$ elements.  It
returns an element $\beta\in\K$ that is normal over $\L$.

We apply Hensel's criterion \cite[Theorem~2.2.4]{Gao93} to the
cyclic extension $\K/\L$. For $a\in\K$, define the resolvent
\[
  \rho_\tau(a)=\sum_{i=0}^{n-1}\tau^i(a)z^i \in \K[z]/(z^n-1).
\]
Then $a$ is normal over $\L$ if and only if
$\rho_\tau(a)$ is a unit in $\K[z]/(z^n-1)$.  

Extend $\tau$ coefficientwise, fixing $z$. Then
\[
  \tau(\rho_\tau(a)) = \sum_{i=0}^{n-1}\tau^{i+1}(a)z^i =
  z^{-1}\sum_{i=0}^{n-1}\tau^i(a)z^i = z^{-1}\rho_\tau(a),
\]
where the middle equality uses $\tau^n(a)=a$ and $z^n=1$.  Let
\[
  u=\rho_\tau(\beta)\in\K[z]/(z^n-1).
\]
Since $\beta$ is normal over $\L$, the element $u$ is a unit.  Apply
$\psi$ coefficientwise to $u$ and multiply in the quotient ring:
\[
  v=\prod_{j=0}^{\ell-1}\psi^j(u)
  \in \K[z]/(z^n-1).
\]
Since $\psi$ has order $\ell$, applying $\psi$ cyclically permutes the
factors in this product.  Hence $v$ is fixed by $\psi$, so
$v\in\E[z]/(z^n-1)$.  Moreover, $v$ is a unit already in this smaller
ring: its inverse in $\K[z]/(z^n-1)$ is also fixed by $\psi$, hence
has coefficients in $\E$.  Write $v=\sum_{i=0}^{n-1}v_i z^i$, where
$v_i\in\E$, and output $\alpha=v_0$.

We now prove $\alpha$ is normal over $\Fq$.  Since
$\tau(u)=z^{-1}u$ and $\tau$ commutes with $\psi$, we have
\[
  \tau(v)=\prod_{j=0}^{\ell-1}\psi^j(\tau(u))
          =\prod_{j=0}^{\ell-1}\psi^j(z^{-1}u)
          =z^{-\ell}v.
\]
Comparing coefficients gives $\tau(v_i)=v_{i+\ell}$,
with indices modulo $n$.  Since $\alpha=v_0$, we have
$v_{r\ell}=\tau^r(\alpha)$ for all $r$.  As $\gcd(\ell,n)=1$, the
indices $r\ell$ run through all residues modulo $n$, and 
\[
  v=\sum_{r=0}^{n-1}\tau^r(\alpha)z^{r\ell}.
\]
Because $\gcd(\ell,n)=1$, the map $z\mapsto z^\ell$ permutes the
monomials $1,z,\ldots,z^{n-1}$ and defines an automorphism of
$\E[z]/(z^n-1)$.  Thus $v$ is the image of
\[
  \rho_\tau(\alpha)=\sum_{r=0}^{n-1}\tau^r(\alpha)z^r
\]
under this automorphism.  Since $v$ is a unit, so is
$\rho_\tau(\alpha)$.  Applying the resolvent test to $\E/\Fq$ with the
generator $\tau|_\E$, we conclude that $\alpha$ is normal over $\Fq$.

This reduction is deterministic. Form $u$ by the iterated Frobenius
algorithm \cite[Algorithm~3.1]{GatSho92}. Since $\psi$ fixes $\theta$,
apply it coefficient-wise using the precomputed $\Fq$-linear map
$c\mapsto c^{q^n}$ on $\L$. Fast cyclic multiplication gives $v$.  As
$\ell=\log^{O(1)}n$, the descent costs $\softO(n^2+n\log q)$
operations in $\Fq$.

\subsection{A complete algorithm for a normal element in any finite
  field}

The preceding construction completes the main theorem for ground
fields of all sizes.
\begin{theorem}
  \label{thm:main}
    Let $\Fq$ be presented as $\FF_p[y]/(\Lambda(y))$, for $\Lambda$
    irreducible in $\FF_p[y]$, let $\Gamma$ be irreducible of degree
    $n$ in $\Fq[\Tau]$ and set $\E=\Fq[\Tau]/(\Gamma(\Tau))$.  There
    exists a deterministic algorithm that outputs a normal element for
    $\E/\Fq$ using
    \[
      O_\epsilon((n^2 \log q)^{1+\epsilon}) + \softO(n\log^2 q)
    \] 
    bit operations, for any fixed $\epsilon > 0$.
\end{theorem}

%%% Local Variables:
%%% mode: LaTeX
%%% TeX-master: "detnorm"
%%% End:

\bibliographystyle{plainurl}
\bibliography{detnorm}

\end{document}